\documentclass[12pt]{amsart}

\author{Paul \textsc{Poncet}}
\address{CMAP, \'{E}cole Polytechnique, Route de Saclay, 91128 Palaiseau Cedex, France \\
and INRIA, Saclay--\^{I}le-de-France}

\email{poncet@cmap.polytechnique.fr}

\usepackage{stmaryrd}
\usepackage{amssymb}
\usepackage{amsmath}
\usepackage{graphicx}
\usepackage{yhmath}
\usepackage{calrsfs} 
\usepackage{times} 

\DeclareMathOperator*{\op}{op}

\newtheorem{theorem}{Theorem}[section]

\newtheorem{proposition}[theorem]{Proposition}
\newtheorem{lemma}[theorem]{Lemma}

\theoremstyle{definition}

\newtheorem{remark}[theorem]{Remark}
\newtheorem{problem}[theorem]{Problem}

\begin{document}

\title{Pruning a poset with veins}

\date{\today}



\begin{abstract}
We recall some abstract connectivity concepts, and apply them to special chains in partially ordered sets, called veins, that are defined as order-convex chains that are contained in every maximal chain they meet. 
Veins enable us to define a new partial order on the same underlying set, called the pruning order. The associated pruned poset is simpler than the initial poset, but irreducible, coirreducible, and doubly-irreducible elements are preserved by the operation of pruning. 
\end{abstract}

\maketitle
 


\section{Introduction}

While every finite semilattice is generated by its irreducible elements, a finite lattice is not always generated by its doubly-irreducible elements. 
For instance, the power set $P_3$ of $\{1, 2, 3\}$ ordered by inclusion is a finite (distributive) lattice with no doubly-irreducible element. 
It turns out that $P_3$ is not a planar lattice, and we know that every finite planar lattice has at least one doubly-irreducible element. 
In the same line, different results were proved that assert the existence of one or several such elements for special types of posets. 

But we need a theorem due to Monjardet and Wille \cite{Monjardet89} augmented by Ern\'e \cite{Erne91} to get necessary and sufficient conditions on a finite distributive lattice to be generated by its doubly-irreducible elements. 

However, the conditions provided by this theorem are related to the normal completion of the lattice at stake, hence seem hardly operational. 
Here we propose a way to ``prune'' a finite poset; this means that we define a new partial order on the same underlying set, called the pruning order. The associated pruned poset is simpler than the initial poset, but irreducible, coirreducible, and doubly-irreducible elements are preserved by the operation of pruning. 
This pruning operation is based on the notion of \textit{vein}, which is an order-convex chain contained in every maximal chain it meets.




\section{Short preliminaries on connectivities}


Here we recall the axiomatic concept of connectivity, which nicely generalizes the corresponding notions used in topological spaces or graphs. 
This will offer an appropriate framework for the study of \textit{veins} in the next section. 
A \textit{connectivity} on a set $E$ is a nonempty collection $\mathrsfs{C}$ of subsets of $E$ covering $E$ and satisfying 
$$
\bigcap \mathrsfs{A} \neq \emptyset \Rightarrow \bigcup \mathrsfs{A} \in \mathrsfs{C}, 
$$
for all subsets $\mathrsfs{A} \subset \mathrsfs{C}$. The elements of $\mathrsfs{C}$ are the \textit{connected} subsets of $E$, and $(E,\mathrsfs{C})$ is called a \textit{connectivity space}. The space is \textit{point-connected} if all singletons are connected. All connectivities considered here will be point-connected.  
The \textit{connected components} of a connectivity space $E$ are the maximal connected subsets. 


We owe this axiomatisation to B\"orger \cite{Boerger84}. Matheron and Serra \cite{Matheron88} and Serra \cite{Serra88, Serra98}, interested in 
applications to mathematical morphology and image analysis, rediscovered this concept, and their work was pursued by Ronse \cite{Ronse98} and Braga-Neto and Goutsias \cite{BragaNeto00, BragaNeto02} among others, for similar purposes. 
At the same time, analogous work arising from order-theoretic interests was developed by Richmond and Vainio \cite{Richmond90} and Ern\'e and Vainio \cite{Erne90b}. 

\section{Irreducible chains in partially ordered sets}

\subsection{Irreducible chains}

A \textit{partially ordered set} or \textit{poset} $(P,\leqslant)$ is a set $P$ equipped with a reflexive, transitive, and antisymmetric binary relation $\leqslant$. 
A nonempty subset $C$ of $P$ is a \textit{chain} (or a \textit{totally ordered subset}) if, for all $x, y \in C$, $x \leqslant y$ or $y \leqslant x$. A chain $M$ is \textit{maximal} if $C \supset M$ implies $C = M$, for all chains $C$ in $P$. 

We call \textit{irreducible} a chain $C$ such that, for all maximal chains $M$, 
$$
C \cap M \neq \emptyset \Longrightarrow C \subset M. 
$$ 
Note that every nonempty subset of an irreducible chain is an irreducible chain. 
The next proposition gives a characterization. 

\begin{proposition}
A chain $C$ is irreducible if and only if, for all nonempty finite (resp.\ arbibrary) families of maximal chains covering $C$, one of them contains $C$. 
\end{proposition}

\begin{proof}
Assume that $C$ is irreducible, and let $(M_j)_{j \in J}$ be some family of maximal chains covering $C$. Let $x \in C$. Then $x \in M_{j_0}$ for some $j_0 \in J$, hence $C \cap M_{j_0} \neq \emptyset$. This implies $C \subset M_{j_0}$. 

Conversely, assume that the property given by the proposition is satisfied for some chain $C$, and let $M$ be a maximal chain meeting $C$ at $x$. If $C$ is not contained in $M$, then $C \cap M^{c} \neq \emptyset$. By Zorn's lemma, there exists some maximal chain $N$ containing $C \cap M^{c}$ and avoiding $x$. Then $C \subset M \cup N$, hence $C \subset N$, a contradiction. 
\end{proof}

Here comes the link with connectivities. 

\begin{proposition}
On a poset, the family of irreducible chains is a connectivity, and maximal irreducible chains correspond to connected components. 
\end{proposition}

\begin{proof}
First notice that every singleton is an irreducible chain. 
Let $(C_j)_{j \in J}$ be a family of irreducible chains with nonempty intersection, and let $M$ be a maximal chain meeting $\bigcup_{j\in J} C_j$ (note that such an $M$ always exists). There is some $j_0 \in J$ such that $C_{j_0} \cap M \neq \emptyset$, so that $C_{j_0} \subset M$. Now for all $j \in J$, $\emptyset \neq C_j \cap C_{j_0} \subset C_j \cap M$, which implies $C_j \subset M$. Therefore, $K = \bigcup_{j\in J} C_j \subset M$, which proves that $K$ is a (nonempty) chain and that this chain is irreducible. 
\end{proof}

\subsection{Veins as irreducible convex chains}

A subset $C$ of a poset is \textit{convex} if, for all $x, y \in C$ with $x \leqslant y$, $[x, y] \subset C$, where the interval $[x, y]$ is the set $\{ z : x \leqslant z \leqslant y\}$. Note that an irreducible chain is not necessarily convex. We define a \textit{vein} as an irreducible convex chain. One can see a vein as a ``path'' with no diversion. 

\begin{proposition}
On a poset, the family of veins is a connectivity, and maximal veins correspond to connected components. 
\end{proposition} 

\begin{proof}
Each singleton is a vein. 
Let $(C_j)_{j \in J}$ be a family of veins with nonempty intersection. We already know that $K = \bigcup_{j\in J} C_j$ is an irreducible chain, let us show that $K$ is convex. So let $x, y \in K$ and $z$ such that $x \leqslant z \leqslant y$. There is some $j_0$ such that $x \in C_{j_0}$ and some $k_0$ with $y \in C_{k_0}$. Take a point $t$ in the intersection of all $C_j$, and let $M$ be a maximal chain containing $\{x, z, y\}$. Since $M$ meets the irreducible chain $K$, it contains $K$. Both $z$ and $t$ are in $M$, so these points are comparable. If $z \leqslant t$, then $z \in [x, t] \subset C_{j_0}$ since $C_{j_0}$ is convex. If $t \leqslant z$, then $z \in [t, y] \subset C_{k_0}$. In either case, $z \in K$, and the convexity of $K$ is proved. 
\end{proof}

\begin{proposition}
Let $P$ be a poset and $Q$ be a subset of $P$. If $V$ is a vein of $P$ meeting $Q$, then $V \cap Q$ is a vein of $Q$. 
\end{proposition}

\begin{proof}
The set $V \cap Q$ is clearly a nonempty convex chain in $Q$. Assume that $M$ is a maximal chain in $Q$ such that $V \cap Q \cap M \neq \emptyset$. Let $N$ be a maximal chain in $P$ containing $M$. Then $V \cap N \neq \emptyset$, so that $V \subset N$ since $V$ is a vein in $P$. This implies that $V \cap Q \subset N \cap Q$. But $N \cap Q$ is a chain in $Q$ containing $M$, so that $M = N \cap Q$ by maximality of $M$. This proves that $V \cap Q \subset M$, i.e.\ that $V \cap Q$ is a vein in $Q$. 
\end{proof}

\subsection{Pruning of a poset}

A vein is \textit{strict} if it is not a singleton. 
On a poset $P$ we can define a new binary relation $\leqslant_{*}$ by $x \leqslant_{*} y$ if $x = y$ or ($x < y$ and there is some maximal chain in $[x, y]$ that contains no strict vein). We call this relation the \textit{pruning order} of $P$. The \textit{pruning} $P^*$ of $P$ is the set $P$ equipped with the pruning order. The following results will justify this wording. 

\begin{theorem}
On every poset, the pruning order is a partial order. 
\end{theorem}

\begin{proof}
The matter is to show the transitivity of $\leqslant_{*}$. Assume that $x \leqslant_{*} y$ and $y\leqslant_{*} z$. If two points among $x, y, z$ are equal, then $x \leqslant_{*} z$, so consider that $x < y < z$. Let $M$ be a maximal chain in $[x, y]$ containing no strict vein, and define $N \subset [y, z]$ similarly. 
Then $M \cup N$ is a chain, and we show that it is maximal in $[x, z]$. So let $C$ be a chain such that $M \cup N \subset C \subset [x, z]$. Then $C \cap [x, y]$ is a chain in $[x, y]$ containing $M$, hence $M = C \cap [x, y]$ by maximality of $M$. Analogously, $N = C \cap [y, z]$. This gives $M \cup N = C \cap ([x, y] \cup [y, z])$. But since $y \in C$, every $c \in C$ is comparable with $y$, so that $C \subset [x, y] \cup [y, z]$. We get $M \cup N = C$, which proves the maximality of $M \cup N$ in $[x, z]$. 

To finish the proof, we show that $M \cup N$ contains no strict vein. Let $V$ be a vein in $M \cup N$, and suppose that we can find some $v, w \in V$ with $v\neq w$ (for instance $v < w$). If both $v, w$ are in $M$, then $[v, w] \subset V$ by order-convexity of $V$, and $[v, w]$ is a strict vein, necessarily contained in $M$, a contradiction. Thus, we must have $v \in M$ and $w \in N$. This gives $v \leqslant y \leqslant w$. Since $v < w$, we can say, for instance, that $v < y$, so that $[v, y] \subset V$ is a strict vein contained in $M$, a contradiction. 
\end{proof}


%


\begin{lemma}\label{lem:starchain}
Let $P$ be a poset, and let $x, y \in P$ such that $x <_{*} y$. 
If $M$ is a maximal chain in $[x, y]$ containing no strict vein, then $M$ is also a chain with respect to the pruning order. 
\end{lemma}

\begin{proof}
Let $x', y' \in M$. Assume for instance that $x' < y'$, and let us prove that $x' <_{*} y'$. This will be the case if we prove that $M' := M \cap [x', y']$ is a maximal chain in $[x', y']$ (containing no strict vein). Let $C$ be a chain such that $M' \subset C \subset [x', y']$, and let $c \in C$. Then $M \cup \{c\}$ satisfies $M \subset M \cup \{c\} \subset [x, y]$. Also, $M \cup \{c\}$ is a chain: if $z \in M$, then $z$ and $c$ are comparable, for either $z < x'$ (in which case $z < c$), or $z > y'$ (in which case $z > c$), or $z \in [x', y']$ (in which case $z \in M'$, hence $z \in C$, and $z$ and $c$ are again comparable as elements of the chain $C$). By maximality of $M$, $M = M \cup \{c \}$, i.e.\ $c \in M$, so that $c \in M \cap [x', y'] = M'$. 
This means that $M' = C$, i.e.\ $M'$ is a maximal chain in $[x', y']$, hence $x' <_{*} y'$. This shows that $M$ is a chain with respect to $\leqslant_{*}$. 
\end{proof}

In a poset $P$, we classically write $x \prec y$ whenever $y$ \textit{covers} $x$, which means that $x < y$ and $[x, y] = \{x, y\}$. 

\begin{lemma}\label{lem:prec}
Let $P$ be a poset, and let $x, y \in P$ such that $x <_{*} y$. If an element $c \in [x, y]$ satisfies $x \prec c$ (resp.\ $c \prec y$), then $x <_{*} c$ (resp.\ $c <_{*} y$). 
\end{lemma}


\begin{proof}
Assume for instance that $c \in [x, y]$ is such that $x \prec c$ (the case $c \prec y$ is similar). 
Note that $C = \{ x, c \}$ is a convex chain. If $C$ is not a vein, then $C$ is a maximal chain in $[x, c]$ containing no strict vein, so that $x <_{*} c$. Suppose on the contrary that $C$ is a vein. 
Since $x <_{*} y$, there is some maximal chain $M$ in $[x, y]$ containing no strict vein. Let $N$ be a maximal chain in $P$ containing $M$. 
Since $x \in C \cap N \neq \emptyset$, we deduce that $C \subset N$, i.e.\ $c \in N$. Thus, $M \cup \{c \}$ is a subchain of $[x, y]$ containing $M$, so that $c \in M$ by maximality of $M$. Then Lemma~\ref{lem:starchain} implies $x <_{*} c$, i.e.\ $C$ is not a vein, a contradiction. 
\end{proof}

\begin{theorem}
Let $P$ be a poset in which every bounded chain is finite. Then $(P^*)^* = P^*$. 
\end{theorem}

\begin{proof}
We use the terms and notations $*$-chain, $*$-vein, $[x, y]_{*}$, etc.\ with obvious definitions. 
Assume that $x <_{*} y$, for some $x, y  \in P$. We want to show that $x <_{**} y$. 
By definition of $<_{*}$, there exists some maximal chain $M$ in $[x, y]$ containing no strict vein. 
By Lemma~\ref{lem:starchain}, $M$ is a (maximal) $*$-chain (in $[x, y]_{*}$). 
To conclude that $x <_{**} y$, it remains to show that $M$ contains no strict $*$-vein. Suppose on the contrary that there is some strict $*$-vein $V$ contained in $M$. With the assumption that every bounded chain in $P$ is finite, we may suppose that $V$ is a two-element $*$-vein, i.e.\ $V = \{ a, b\}$ with $a <_{*} b$. 

Let us show that $V$ is convex. So let 
$c \in P$ such that $a \leqslant c \leqslant b$. 
Since $a <_{*} b$, there is some maximal chain $N$ in $[a, b]$ containing $c$. 
We assumed that every bounded chain in $P$ is finite, so 
we can write $N$ as $a = n_0 < n_1 < \ldots < n_m = b$, and $n_k = c$ for some $c$. By maximality of $N$, we have $n_0 \prec n_1$, so that $n_0 <_{*} n_1$ by Lemma~\ref{lem:prec}. If $N^*$ is a maximal $*$-chain containing $\{n_0, n_1\}$, then $a \in V \cap N^* \neq \emptyset$, so $V \subset N^*$ since $V$ is a $*$-vein. This implies that $b \in N^*$, so either $b \leqslant_{*} n_1$ or $n_1 \leqslant_{*} b$. But we also know that $n_1 \leqslant b$, so $n_1 \leqslant_{*} b$. We see that $n_1 \in [a, b]_{*}$; since $V$ is $*$-convex, this proves that $n_1 \in V$. We deduce by induction that $n_j \in V$ for all $j$, so in particular $c \in V$, and we have shown that $V$ is convex. 

Now let us show that $V$ is irreducible. So let $M'$ be a maximal chain in $P$ such that $V \cap M' \neq \emptyset$. We want to show that $V \subset M'$. 
We may suppose, without loss of generality, that $a \in V \cap M'$. 
The hypothesis made on $P$ implies the existence of some $\beta \in M'$ such that $a \prec \beta$. 
Then $M'' = \{ a, \beta \}$ is a maximal chain in $[a, \beta]$. 

First case: $M''$ contains a strict vein. Then $M''$ is itself a vein. If $N''$ is a maximal chain containing $\{a, b\}$, then $a \in M'' \cap N'' \neq \emptyset$, so that $M'' \subset N''$. Thus, $\beta \in N''$, so $\beta$ and $b$ are comparable. 

Second case: $M''$ contains no strict vein. 
Then $a <_{*} \beta$. Now if $N^*$ is a maximal $*$-chain containing $\{a, \beta \}$, then $V \cap N^* \neq \emptyset$. Since $V$ is a $*$-vein, this implies $V \subset N^*$, so $b \leqslant_{*} \beta$ or $\beta \leqslant_{*} b$. Again, $\beta$ and $b$ are comparable. 

Since $V$ is convex and $a \prec \beta$, both cases imply that $b = \beta$. 
So we have $b \in M'$, i.e.\ $V \subset M'$, which shows that $V$ is irreducible. 

We have proved that $V$ is a strict vein contained in $M$, a contradiction. So $M$ contains no strict $*$-vein, and $x <_{**} y$. 

Conversely, if $x <_{**} y$, then $x <_{*} y$ is obvious, so we have proved that $x \leqslant_{**} y \Leftrightarrow x \leqslant_{*} y$ for all $x, y \in P$, i.e.\ $(P^*)^* = P^*$. 
\end{proof}

\begin{remark}
If $P$ contains an infinite chain, we may have $(P^*)^* \neq P^*$. Consider for instance $P = [0, 1] \cup \{\omega\}$, where $\omega$ is an additional element such that $0 < \omega < 1$. Then, in $P^*$, no relation holds but $0 <_{*} \omega <_{*} 1$ and, in $(P^*)^*$, no elements are comparable. 
\end{remark}

\begin{problem}
Is it true that $((P^*)^*)^* = (P^*)^*$ for every poset $P$?
\end{problem}

In a poset, an element $x$ is \textit{irreducible} if $x$ is a maximal element or $\uparrow\!\! x \setminus \{x\}$ is a filter, \textit{coirreducible} if it is irreducible in $P^{\op}$, and doubly-irreducible if it is both irreducible and coirreducible. 
Remark that if $P$ is conditionally complete, then $x$ is irreducible if and only if $x = a \wedge b$ implies $x \in \{a, b\}$, for all $a, b$, where $a \wedge b$ denotes the infimum of $\{a, b\}$. 

\begin{figure}
	\centering
		\includegraphics[width=0.9\textwidth]{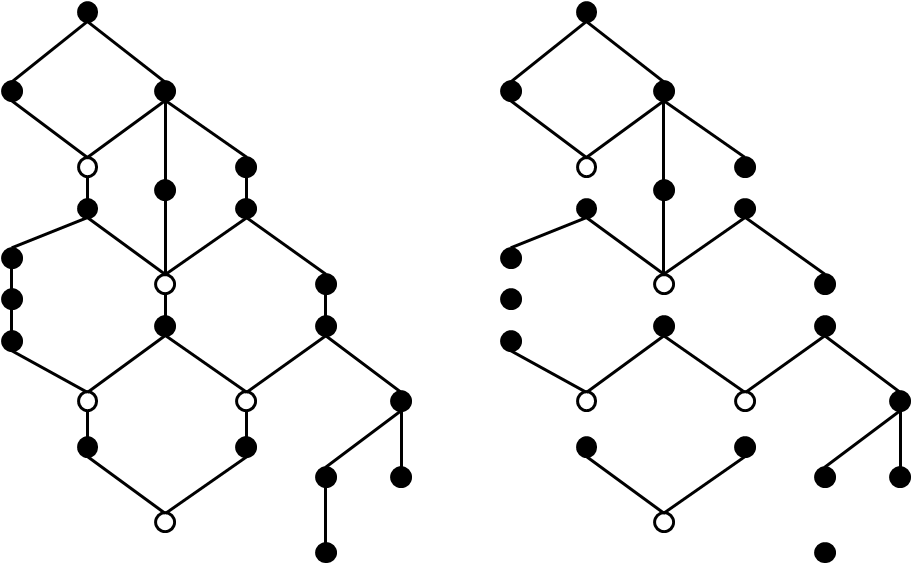}
	\caption{ On the left, a poset $P$ with nineteen irreducible elements (in black); on the right, the pruned poset $P^*$ has 
	the same irreducible elements as $P$. }
	\label{fig:AppendixPrunedPoset}
\end{figure}


\begin{proposition}
Let $P$ be a finite conditionally complete poset and $x \in P$. Then 
\begin{itemize}
	\item $x$ is irreducible in $P$ if and only if $x$ is irreducible in $P^*$, 
	\item $x$ is coirreducible in $P$ if and only if $x$ is coirreducible in $P^*$.  
\end{itemize}
\end{proposition}

\begin{proof}
Let $x \in P$ and assume that $x$ is not irreducible. Then there are $a, b$ such that $x = a \wedge b$ and $x \notin \{a, b\}$, and we can assume that $x \prec a$ and $x \prec b$ since $P$ is finite. Then $\{x, a \}$ is a maximal chain in $[x, a]$. Moreover, it contains no strict vein: if $V$ is a strict vein included in $\{x, a\}$, then $V = \{x, a\}$; but if $M$ is a maximal chain containing $\{x, b\}$, then $V \cap M = \{x\} \neq \emptyset$, while $V \not\subset M$. Hence $x \leqslant_{*} a$, and symmetrically $x \leqslant_{*} b$. Moreover, if $u$ satisfies $u \leqslant_{*} a$ and $u \leqslant_{*} b$, then $u \leqslant a$ and $u \leqslant b$, so that $u \leqslant a \wedge b = x$. This shows that $x$ is the infimum in $P^*$ of $\{a, b\}$, so $x$ is not irreducible in $P^*$. 

Conversely, let $x$ be irreducible in $P$, and let us show that $x$ is irreducible in $P^*$. So let $a, b$ such that $x <_{*} a$ and $x <_{*} b$. This implies that $x \leqslant a\wedge b$, and even $x  < a\wedge b$ since $x$ is irreducible in $P$. 
Let $c$ such that $x \prec c \leqslant a \wedge b$. We show that $c \leqslant_{*} a$. Since $x <_{*} a$, there is a maximal chain $M$ in $[x, a]$ containing no strict vein. Considering that $x \prec c$, we see that $M \setminus \{x\}$ is a maximal chain in $[c, a]$ containing no strict vein. Consequently, $c \leqslant_{*} a$. Similarly, $c \leqslant_{*} b$. 
This proves that the subset $\{ a \in P : x <_{*} a \}$ is either empty or filtered, i.e.\ that $x$ is irreducible in $P^*$. 
\end{proof}



%



\section{Conclusion and perspectives}

A future work may consist in finding an algorithm to efficiently prune a given poset.

\bibliographystyle{plain}

\def\cprime{$'$} \def\cprime{$'$} \def\cprime{$'$} \def\cprime{$'$}
  \def\ocirc#1{\ifmmode\setbox0=\hbox{$#1$}\dimen0=\ht0 \advance\dimen0
  by1pt\rlap{\hbox to\wd0{\hss\raise\dimen0
  \hbox{\hskip.2em$\scriptscriptstyle\circ$}\hss}}#1\else {\accent"17 #1}\fi}
  \def\ocirc#1{\ifmmode\setbox0=\hbox{$#1$}\dimen0=\ht0 \advance\dimen0
  by1pt\rlap{\hbox to\wd0{\hss\raise\dimen0
  \hbox{\hskip.2em$\scriptscriptstyle\circ$}\hss}}#1\else {\accent"17 #1}\fi}






\end{document}